\documentclass[11pt, letterpaper, twocolumn]{article}
\usepackage{cite}
\usepackage{graphicx}
\usepackage{lipsum}
\usepackage{multicol}
\usepackage{graphicx}
\usepackage{caption}
\usepackage{refstyle}
\usepackage{amsmath}
\usepackage{amsthm}
\usepackage{amsmath}
\usepackage{amsfonts}
\usepackage{amssymb}
\usepackage{bbm}
\usepackage{amsmath,array}
\usepackage{abstract}
\usepackage{mathtools}
\usepackage{slashbox}
\usepackage[all]{xy}
\usepackage{mathtools}
\usepackage[utf8]{inputenc}
\usepackage[english]{babel}
\usepackage[T3,T1]{fontenc}
\usepackage{fancyhdr}
\usepackage{graphicx} 
\DeclareSymbolFont{tipa}{T3}{cmr}{m}{n}
\DeclareMathAccent{\invbreve}{\mathalpha}{tipa}{16}

\usepackage[left=1in,right=1in,top=1in,bottom=1in]{geometry}
\usepackage{setspace}
\newtheorem{theorem}{Theorem}

\begin{document}

\date{}

\title{Randomized Key Encapsulation/Consolidation} 
\author{Amir K. Khandani \\ E\&CE Dept., Univ. of Waterloo, Waterloo, Ontario, Canada; khandani@uwaterloo.ca}

 \maketitle

\noindent
\underline{\bf Abstract:}This article bridges the gap between two topics used in sharing an encryption key: (i) Key Consolidation, i.e., extracting two identical strings of bits from two information sources with similarities (common randomness).  (ii) Quantum-safe Key Encapsulation by incorporating randomness in Public/Private Key pairs.  In the context of Key Consolidation, the proposed scheme adds to the complexity Eve faces in extracting useful data from leaked information. In this context, it is applied to the method proposed in~\cite{AK1} for establishing common randomness from round-trip travel times in a packet data network.   The proposed method allows adapting the secrecy level to the amount of similarity in common randomness. It can even encapsulate a Quantum-safe encryption key in the extreme case that no common randomness is available. In the latter case, it is shown that the proposed scheme offers improvements with respect to the McEliece cryptosystem which currently forms the foundation for Quantum safe key encapsulation. 

\section{Introduction}
Traditional  methods used in sharing an encryption key  rely on mathematical tools to construct a one-way function that is hard to invert. In particular, in Public Key Infrastructure (PKI), data is embedded in a (publicly sent) message using a public key, which is then extracted using its corresponding private key. Typically,  PKI is used to establish (encapsulate) an encryption key which is then used with Advanced Encryption System (AES). To safeguard key encapsulation against Quantum computers, methods considered for Quantum-safe encryption incorporate randomness in the public/private key pairs.  

Key consolidation, motivated by existence results in \cite{common}, concerns extracting a secret key from dependent random variables. This is of interest in areas such as Quantum Key Distribution (QKD) and Physical Layer Security (PLS). Method proposed in~\cite{B1} for key consolidation is widely used in the context of QKD, however, it requires extensive back-and-forth public communications between legitimate parties. Later works have studied the use of Turbo-codes~\cite{B2}\cite{B3} and Low Density Parity Check 
codes ~\cite{B4} to \cite{B7}  to improve upon~\cite{B1}. These earlier works suffer from:  (i) information leakage which is not yet rigorously quantified/studied, (ii) failing to function in the absence of  common randomness, or when the common randomness is of poor quality.  The current article aims to address these shortcomings. It is the first work to merge the two areas of ``quantum-safe (randomized) encryption" and ``key consolidation''. The proposed method is accompanied by information theoretical proofs. These proofs guarantee a target security level, called $\mathsf{SEC}$ (typically 256 bits), is realized, where the only possible attack is the one based on an exhaustive search over a set with $2^{\mathsf{SEC}}$ elements.

\section{Proposed Structure} 

\subsection{Construction and Security Level}
Alice generates the public key $\mathbf{P}=\mathbf{BC}$ as shown in Fig.~\ref{PK1}. Matrix   $\mathbf{C}_1$ is  a punctured random permutation matrix. Puncturing is performed by randomly selecting $\mathsf{p+q}$ columns, indexed by 
$i_p\in[1,\mathsf{s=m+p+q}]$ for $p=1,\ldots, \mathsf{p+q}$, 
setting column indexed by $i_p$ equal to zero, and then shifting rows indexed by 
$i_p+1,\ldots,\mathsf{s}$ to positions $i_p,\ldots,\mathsf{s}-1$, respectively, and discarding the last (repeated) rows. This results in matrix $\mathbf{C}_1$ of size $\mathsf{s}\times \mathsf{m}$ (see Fig.~\ref{PK1}). Matrix $\mathbf{C}$ is formed from $\mathbf{C}_1$ by appending a random matrix $\mathbf{C}_2$ of size $\mathsf{p}\times \mathsf{s}$. Matrix  $\mathbf{B}$ is constructed as depicted in Fig.~\ref{RC1}, 
and captured in expressions \ref{RC00} to \ref{RC003}. 

\begin{figure}[h]
   \centering
\hspace*{-0.8cm}
   \includegraphics[width=0.461\textwidth]{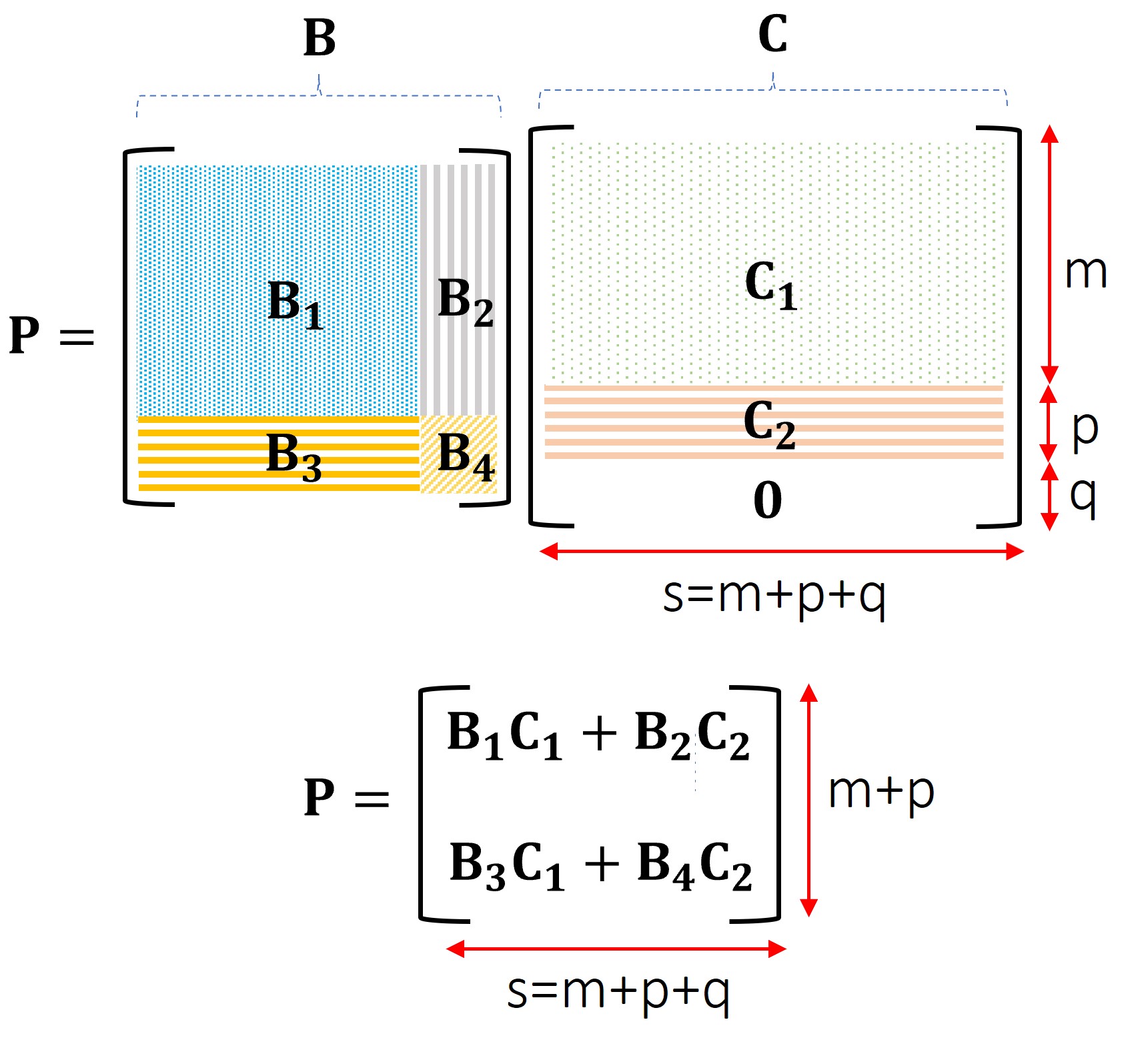}
   \caption{Structure of the public key: $\mathbf{C}_1$ is formed 
by puncturing $\mathsf{p+q}$  columns
 from an $\mathsf{s}\times \mathsf{s}$ random permutation matrix where $\mathsf{s}=\mathsf{m+p+q}$, and $\mathbf{C}_2$ is an 
$\mathsf{p}\times \mathsf{s}$ random matrix.  }
   \label{PK1}
 \end{figure}

%

 \begin{figure}[h]
   \centering
\hspace*{-0.8cm}
   \includegraphics[width=0.48\textwidth]{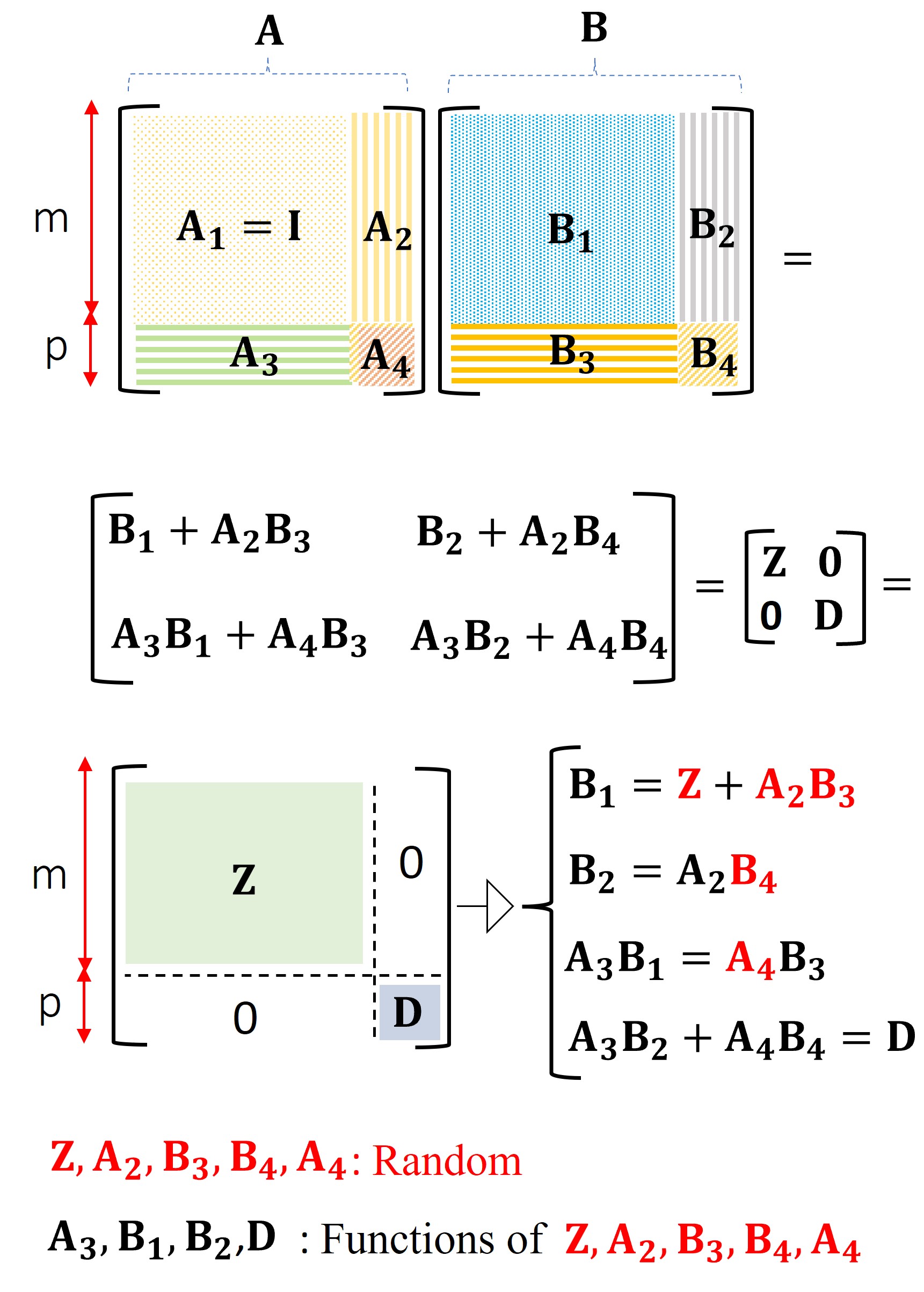}
   \caption{Structure of the (private to Alice) matrix $\mathbf{A}$ and relevant randomness conditions necessary for key recovery by Alice.  }
   \label{RC1}
 \end{figure}

Conditions in Fig.~\ref{RC1} are:
\begin{eqnarray}  \nonumber
\mathbf{AB}  & = &
\begin{bmatrix}
\mathbf{I}\mathbf{B}_1+\mathbf{A}_2\mathbf{B}_3 & 
\mathbf{I}\mathbf{B}_2+\mathbf{A}_2\mathbf{B}_4 \\
\mathbf{A}_3\mathbf{B}_1+\mathbf{A}_4\mathbf{B}_3 & 
\mathbf{A}_3\mathbf{B}_2+\mathbf{A}_4\mathbf{B}_4
\end{bmatrix} \\  \nonumber
&  & \\  \label{RC00}
& = &
\begin{bmatrix}
\mathbf{Z} & \mathbf{0}\\
\mathbf{0} & \mathbf{D}
\end{bmatrix}. 
\end{eqnarray}
where $\mathbf{Z}$ is a permutation matrix. 
We obtain:
\begin{eqnarray} \label{RC000}
\mathbf{B}_1 & = & \mathbf{Z}+\mathbf{A}_2\mathbf{B}_3 \\ \label{RC001}
\mathbf{B}_2 & = & \mathbf{A}_2\mathbf{B}_4 \\ \label{RC002}
\mathbf{A}_3\mathbf{B}_1 & = & \mathbf{A}_4\mathbf{B}_3 \\ \label{RC003}
\mathbf{A}_3\mathbf{B}_2+\mathbf{A}_4\mathbf{B}_4 & = & \mathbf{D}.
\end{eqnarray}
Noting $\mathsf{m}>\mathsf{p}$, from \ref{RC000},  $\mathbf{A}_2$ and $\mathbf{B}_3$ can be selected to be random, creating a partial randomness in  $\mathbf{B}_1$. Then, in \ref{RC001}, for given $\mathbf{A}_2$,  matrix $\mathbf{B}_4$ can be selected  randomly, creating the randomness in  matrix  $\mathbf{B}_2$. Matrix  $\mathbf{A}_4$ in 
\ref{RC002} can be selected randomly  for given 
$\mathbf{B}_1$, $\mathbf{B}_3$, while finding $\mathbf{A}_3$ to satisfy 
\ref{RC002}. Finally, $\mathbf{D}$ in \ref{RC003} is computed as a function of 
$\mathbf{A}_3$, $\mathbf{B}_2$, $\mathbf{A}_4$ and $\mathbf{B}_4$.

\begin{theorem} \label{Th1}
Random permutation embedded in $\mathbf{C}_1$ is completely masked by randomness embedded in $\mathbf{B}$. 

\end{theorem}
\begin{proof}
Given a realization of matrix $\mathbf{C}_1$, let us consider the matrix 
$\mathbf{X}\mathbf{C}_1$ where $\mathbf{X}$ is a block diagonal matrix composed of $\mathsf{m}\times \mathsf{m}$ random permutation matrix $\mathbf{Y}$ in its upper left corner, and $\mathsf{p}\times \mathsf{p}$ identity matrix in its lower right corner. 
Given  $\mathbf{P}=\mathbf{BC}$, let 
 us consider a different realization of  $\mathbf{P}$ as $\mathbf{P}=\mathbf{BXC}$. It follows that the permutation in rows of $\mathbf{C}_1$ due to $\mathbf{Y}$ can be absorbed in matrix $\mathbf{B}$, by permuting its columns  as follows 
\begin{equation}
\begin{bmatrix}
\mathbf{B}_1 \\
\mathbf{B}_3
\end{bmatrix}~~\rightarrow~~
\begin{bmatrix}
\mathbf{B}_1 \\
\mathbf{B}_3
\end{bmatrix} \mathbf{Y}= 
\begin{bmatrix}
\mathbf{B}_1\mathbf{Y} \\
\mathbf{B}_3 \mathbf{Y}
\end{bmatrix}.
\label{Eqp0}
\end{equation}
Noting $\mathbf{B}_3$ is random, $\mathbf{B}_3 \mathbf{Y}$ will be another random realization of the underlying matrix that could occur  with the same probability.  This means, in information theoretical sense, the effect of $ \mathbf{Y}$  in
$\mathbf{B}_3 \mathbf{Y}$ cannot be distinguished by Eve. 
In other words, for $\mathbf{B}_1 \mathbf{Y}$, replacing in~\ref{RC000}, it follows 
\begin{equation}
\mathbf{B}_1\mathbf{Y}  =  \mathbf{ZY}+\mathbf{A}_2\mathbf{B}_3 \mathbf{Y}=
\hat{\mathbf{Z}}+\mathbf{A}_2\hat{\mathbf{B}}_3,
\label{Eqp00}  
\end{equation}
where $\hat{\mathbf{Z}}$ and $\hat{\mathbf{B}}_3$ are realizations that could occur with the same probability as $\mathbf{Z}$ and $\mathbf{B}_3$.

A similar conclusion  could be reached relying on a different perspective. Let us consider 
matrix $\mathbf{B}_1\mathbf{C}_1+\mathbf{B}_2\mathbf{C}_2$ (as part of the public matrix in Fig.~\ref{PK1}) and study the maximum information it could provide about $\mathbf{C}_1$. Without loss of generality, let us act to the benefit of Eve by assuming $\mathbf{B}_1$ is available to Eve as side information. In this case, $\mathbf{B}_2\mathbf{C}_2$ acts as a mask hiding $\mathbf{C}_1$. 
For values of parameters $\mathsf{m}$, $\mathsf{p}$, $\mathsf{s}$ considered here, the information content of $\mathbf{B}_2\mathbf{C}_2$ (computed in \cite{Arxiv}) is much higher than the target $\mathsf{SEC}$ level, entailing  $\mathbf{C}_1$ is well hidden.
\end{proof}
\begin{theorem} \label{Th2}
Matrix $\mathbf{B}_3\mathbf{C}_1+\mathbf{B}_4\mathbf{C}_2$ does not provide any information about matrix $\mathbf{C}_1$.
\end{theorem}
\begin{proof}
Assume $\mathbf{B}_3$ is available to Eve as side information. Noting $\mathbf{B}_4$ is a  random $\mathsf{p}\times \mathsf{p}$ matrix independent of $\mathbf{C}_2$, it follows that $\mathbf{B}_4\mathbf{C}_2$ will be a random $\mathsf{p}\times \mathsf{p}$ matrix,  completely hiding $\mathbf{B}_3\mathbf{C}_1$. Value of $\mathsf{p}$ is selected such that the information content of $\mathbf{B}_4$ is higher than the target $\mathsf{SEC}$ level. 
\end{proof} 
\subsection{Encryption and Decryption}
{\bf Encryption:} Bob selects a binary vector 
$\mathbf{c}=[\mathbf{f}_1,  \ldots, \mathbf{f}_{\mathsf{r}}]^t$ where 
$\mathbf{f}_i$, $i=1,\ldots,\mathsf{r}$, called {\em component codes}, hereafter, are code-words from a short forward error correcting code of length $\ell=\mathsf{s}/\mathsf{r}$ composed of $\mathsf{f}$ code-words. 
A key $\mathbf{k}$ of size $\mathsf{k}=\lfloor \mathsf{r}\log_2(\mathsf{f})\rfloor$ bits is encapsulated in $\mathbf{c}_{\mathbf{k}}$. Code-words of the $i$'th component code are mapped to integers $0,\ldots, \mathsf{f}$ relying on a random assignment for each component code. As a result, Eve will not be able to rely on a generator matrix to map the labels to the code-words of the components codes. Relying on a generator matrix would provide an isomorphism between labels and code-words, which would enable Eve to benefit from information set decoding attack applied to a smaller set of labeling bits. This could be realized by forming 
$\mathbf{PG}$ where $\mathbf{G}$ is a block diagonal matrix generating concatenation of component codes. Each non-isomorphic labeling can be summarized by a different ordering of integers $0,\ldots, \mathsf{f}-1$ which is selected by Alice and publicly shared with Bob. 
Then, the public message $\mathbf{m}_{\mathbf{k}}$ 
is formed as
\begin{equation} 
\mathbf{m}_{\mathbf{k}} = \mathbf{P}\left(\mathbf{c}_{\mathbf{k}}+\mathbf{e}_1+
\mathbf{r}_1\right)+\mathbf{e}_2+\mathbf{r}_2
\label{Eqp110}
\end{equation}
where $\mathbf{e}_1$, $\mathbf{e}_2$ are added by Bob as error vectors 
and $\mathbf{r}_1$, $\mathbf{r}_2$ are bits from common randomness at the Bob's side. Positions of common random bits in $\mathbf{r}_1$ 
and $\mathbf{r}_2$ are publicly known, and 
positions of ones in $\mathbf{e}_1$ and $\mathbf{e}_2$ are randomly selected by Bob.    
Number of bits forming $\mathbf{r}_1$ 
and $\mathbf{r}_2$ depend on the total number of bits
extracted from common randomness. Number of ones in $\mathbf{e}_1$ and $\mathbf{e}_2$ are selected such that the overall error vector added to each component code is within its error correction capability.
 Denoting the number of ones in $\mathbf{e}_1$, $\mathbf{e}_2$ as $w_{\mathbf{e}_1}$, $w_{\mathbf{e}_2}$, respectively, the total number of random errors added to each component code will be equal to 
$w_{\mathbf{e}_1}+w_{\mathbf{e}_2}$.  Note that the multiplication 
by $\mathbf{A}$ at Alice's side will keep bits corresponding to each component code within its original boundary.   Typically, $w_{\mathbf{e}_1}+w_{\mathbf{e}_2}$ is fixed while values of $w_{\mathbf{e}_1}$ and 
$w_{\mathbf{e}_2}$ are (randomly)  selected by Bob (for each component code). Note that, in~\ref{Eqp110}, by adding vectors of errors and common randomness to both  $\mathbf{c}_{\mathbf{k}}$ and $\mathbf{P}\left(\mathbf{c}_{\mathbf{k}}+\mathbf{e}_1+\mathbf{r}_1\right)$, their effects propagates in vectors that Eve can observe by working directly on \ref{Eqp110}, or by aiming to invert $\mathbf{P}$ in ~\ref{Eqp110}. Noting randomness in $\mathbf{P}$, this makes the information set decoding attack more complex (due to error propagation) as compare to an exhaustive search attack.  

\noindent 
{\bf Decryption:} 
Alice, having access to the private key $\mathbf{A}$, can form
\begin{equation}
\mathbf{A} \mathbf{P}=
 \mathbf{A} \mathbf{B}\mathbf{C}=
\begin{bmatrix}
 \mathbf{Z}  & \mathbf{0} \\
\mathbf{0} & \mathbf{D}
\end{bmatrix}
\begin{bmatrix}
    \mathbf{C}_1    \\
\mathbf{C}_2
\end{bmatrix}=
\begin{bmatrix}
    \mathbf{Z} \mathbf{C}_1    \\
\mathbf{D}\mathbf{C}_2
\end{bmatrix}.
\end{equation}
Upon receiving $\mathbf{m}_{\mathbf{k}}$ from Bob, Alice computes $\mathbf{A}\mathbf{m}_{\mathbf{k}}$ and discards the last $\mathsf{p}$ bits (corresponding to $\mathbf{D}\mathbf{C}_2$) from the result. 
Noting the structure of $\mathbf{A}$ in Fig.~\ref{RC1}, the first 
$\mathsf{m}$ bits in the product $\mathbf{A}\mathbf{m}_{\mathbf{k}}$ are obtained from the first  $\mathsf{m}$ bits in 
$\mathbf{c}_{\mathbf{k}}+\mathbf{e}_1+\mathbf{r}_1$. These bits are permuted according to block permutation matrix $\mathbf{Z}\mathbf{C}_1$ (which keeps the bits corresponding to each component code within its original boundary), and then are added to $\mathbf{e}_2+\mathbf{r}_2$. The final binary vector satisfies:   
(i) Bits corresponding to each component code remain within 
its respective boundaries.  (ii) Alice is aware of the positions of punctured bits as well as the positions of common random bits. (iii) Number of added errors in each component code is within its respective error correction capability. As a result, even in the absence of common random bits, Alice will be able to correct all errors and recover the key. 

\noindent {\bf Attack:} Noting above, the only option for attack is to exhaustively examine all binary vectors formed by concatenation of component codes, multiply each resulting vector by $\mathbf{P}$, compare each outcome with  $\mathbf{m}_{\mathbf{k}}$, count the number of errors and decide if the encapsulated key is found. Number of components codes, $\mathsf{r}$,  and their respective number of code-words,
$\mathsf{f}$, are selected such that the resulting set is of size 
$2^\mathsf{SEC}$.

\noindent 
{\bf Reed Muller Component Codes:}
 In this section, code-words of component codes are selected from Reed Muller codes $(16,5,8)$ or $(32,6,16)$ upon discarding all-zeros and all-ones code-words. This results in $30$ and $62$ code-words, respectively. 
Matrix $\mathbf{C}_1$ is constructed such that a single element, in a random position, is punctured from each component code. The minimum distance for the resulting punctured codes is equal to $7$ and $15$, capable of correcting $3$ and $7$ bit errors, respectively. Bob selects $w_1$ and $w_2$ such that $w_1+w_2= 3$ and $7$, respectively. 
To encapsulate a key of size 256 bits, numbers of component codes are selected as
\begin{eqnarray} \label{EQP1}
\lceil 256/\log_2(30)\rceil & = & 53 \\ \label{EQP2}
\lceil 256/\log_2(62)\rceil& = & 43
\end{eqnarray}
for $(16,5,8)$ and $(32,6,16)$, resulting in message size of 
$\mathsf{m}=53\times 16=848$ and $\mathsf{m}=43\times 16=1376$, respectively. Value of  $\mathsf{p}$ and $\mathsf{q}$ for $(16,5,8)$, $(32,6,16)$ are selected as 
$\mathsf{p}=27,22$ and $\mathsf{q}=26,21$, respectively. The resulting 
$(\mathsf{m+p})\times (\mathsf{m+p+q})$ public key matrices are of sizes 
\begin{eqnarray} \label{EQQ1}
(848+27)\times (848+53) & \approx & 0.8 \\  \label{EQQ2} 
(1376+22)\times (1376+43) & \approx &  2
\end{eqnarray}
mega-bits, respectively. For Eve to locate the erroneous positions, the $\log_2$ of the number of positions to be exhaustively searched will be equal to
\begin{eqnarray} \label{EQR1}
53\times \log_2{15\choose 3} & \approx & 468 \\  \label{EQR2}
43\times \log_2{31\choose 7}& \approx & 917
\end{eqnarray}
respectively, which are higher than the target security level of $256$ bits. 
Matrix $\mathbf{C}_1$ is selected such that a single random position  within each component code is punctured, and bits forming each component code are permuted among themselves.  
For $(16,5,8)$, $(32,6,16)$ RM components codes, upon discarding all-zeros and all-ones code-words, remaining code-words each include an equal number of zeros and ones.  Consequently, the number of possibilities for puncturing a single bit and permuting the remaining  $15,31$ positions are equal to $\frac{16\times 15!}{8!\times 7!}$ and $\frac{32\times 31!}{16!\times 15!}$, resulting in security levels of 
\begin{eqnarray}\label{EQP3}
 \frac{16\times 15!}{8!\times 7!} \rightarrow  
\log_2\left(\frac{16!}{8!\times 7!}\right)=16.65 \\  \label{EQP4}
\frac{32\times 31!}{16!\times 15!} \rightarrow  
\log_2\left(\frac{32!}{16!\times 15!}\right)=33.16
\end{eqnarray}
bits, respectively. Multiplying \ref{EQP1} by  \ref{EQP3} and \ref{EQP2} by  \ref{EQP4}, it is concluded that the resulting total security levels
are significantly higher than the target security level of $256$ bits.

\subsection{Complexity Comparisons} \label{compAK}
Complexity is compared to that of McEliece/Niederreiter  cryptosystem \cite{R0} using an $(n,k)$ Goppa code (extracted from a proposal submitted to National Institute of Standards and Technology \cite{Main-ref}).
Complexity aspects include: (i) Storage requirement for storing the public key. (ii) Computational complexity of key encapsulation and recovery. The main computational complexity in McEliece/Niederreiter 
cryptosystems concerns decoding of the underlying code, while decoding of the components codes in our case is fairly simple.  For this reason, our comparisons do not include the decoding complexity. This omission results in disadvantages for the proposed method.  
Table~\ref{TabN1} provides size of the public key for a security level of 256.  
  
\begin{table}
\begin{small}
\begin{center}
\begin{tabular}{|c|c|c|c|} \hline
McEliece & $n$ & $k$  &  Memory  \\ \cline{2-4}
Cryptosystem & 6624 & 5129  & 7.6 Mbits \\ \hline \hline
Proposed Method  & 16 & 5 & 0.8  Mbits\\ \hline
Proposed Method  & 32 & 6 & 2.0  Mbits \\ \hline
\end{tabular} 
\caption{Key sizes for $\mathsf{SEC}=256$ bits. 
For McEliece, memory is based on  a systematic generator occupying $k(n-k)$ bits, and $\mathsf{SEC}$ is based on information set decoding attack~\cite{Main-ref}.}
\label{TabN1}
\end{center}
\end{small}
\end{table}

\subsection{Example for Key Consolidation} 

Reference~\cite{AK1} presents a new method for establishing common randomness between a Node $\mathsf{A}$ and a Node $\mathsf{B}$ over the Internet. Node 
$\mathsf{A}$ sends a sequence of $N$ User Datagram Protocol (UDP) packets at regular time intervals (typically with 10\,msec time gap) to Node $\mathsf{B}$. Node $\mathsf{B}$ sends the received packets, one by one, back to Node $\mathsf{A}$  and Node $\mathsf{A}$  sends them back to Node $\mathsf{B}$. Figure~\ref{fig2} shows an example for $L=2$ loops. Due to looping, these round-trip times, although random, will be close to each other (will have $2L-2$ common travel times).  Then, Node $\mathsf{A}$ and Node $\mathsf{B}$ separately measure their corresponding total round trip times for each packet, and assign a zero/one to each packet if the corresponding round trip time is smaller/larger than the mean of travel times measured at the respective node. The extracted bits are used as common randomness.   Reference~\cite{AK1} presents a method for Bob to extract soft information about  the bits at its end, which are then utilized to correct errors between its local copy of common random bits and that of Alice. Using the method of current article, Alice embeds its version of common random bits in a public message sent to Bob. Figure~\ref{fig3} shows examples  of the achieved performance, where all bits in each RM code are added to bits extracted from common randomness. The two nodes are deployed on Microsoft Azure between North America and Europe.
  \begin{figure}[htbp]
   \centering
   \includegraphics[width=0.4\textwidth]{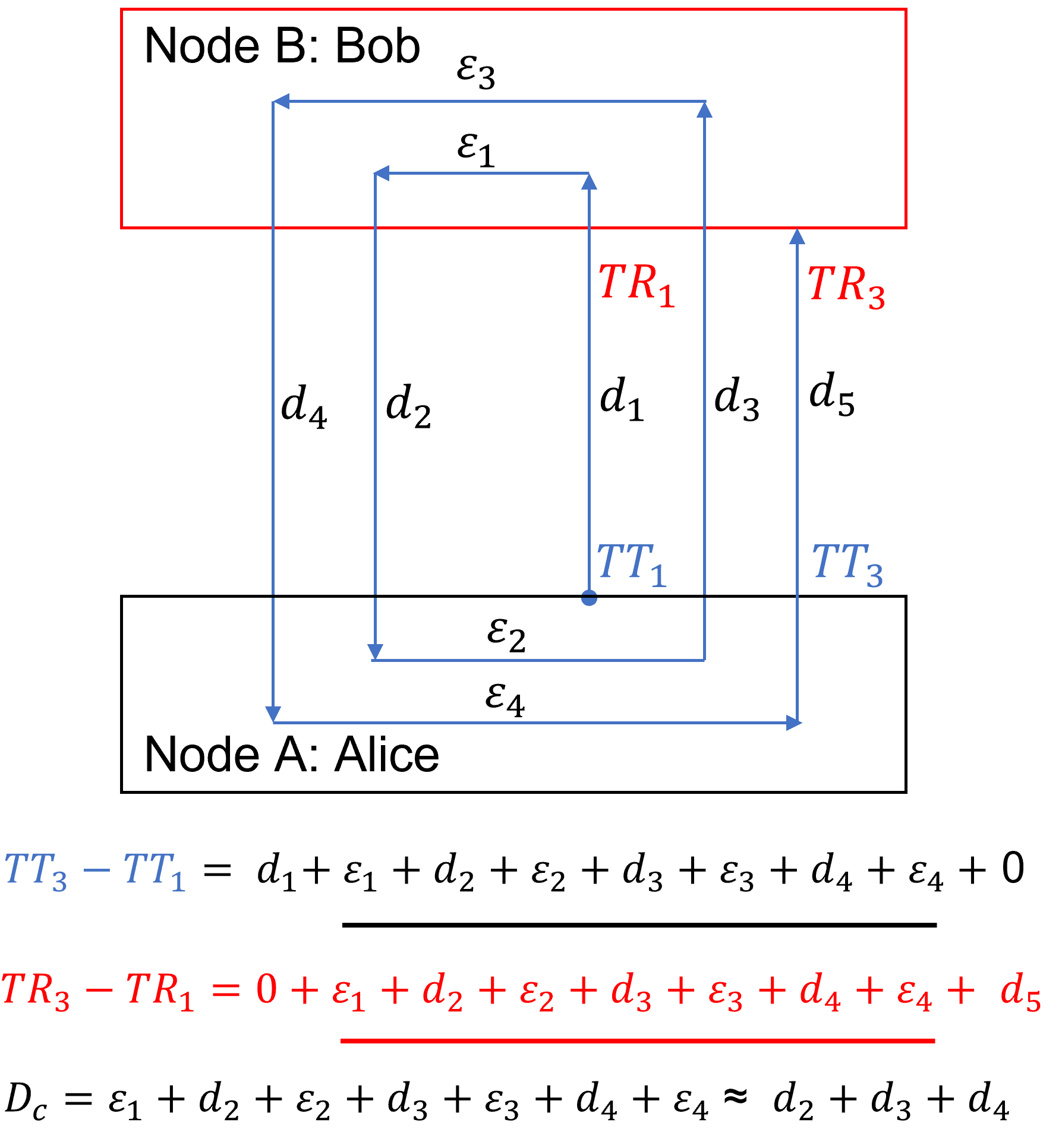}
   \caption{
Round trip time, namely $TT_3-TT_1$ and $T\!R_3-T\!R_1$,  are dependent random variables (since  $\epsilon_1+d_2+\epsilon_2+d_3+\epsilon_3+d_4+\epsilon_4$ is in common). }
 \label{fig2} 
\end{figure} 
  \begin{figure}[htbp]
   \centering
  \hspace*{-0.25cm} \includegraphics[width=0.53\textwidth]{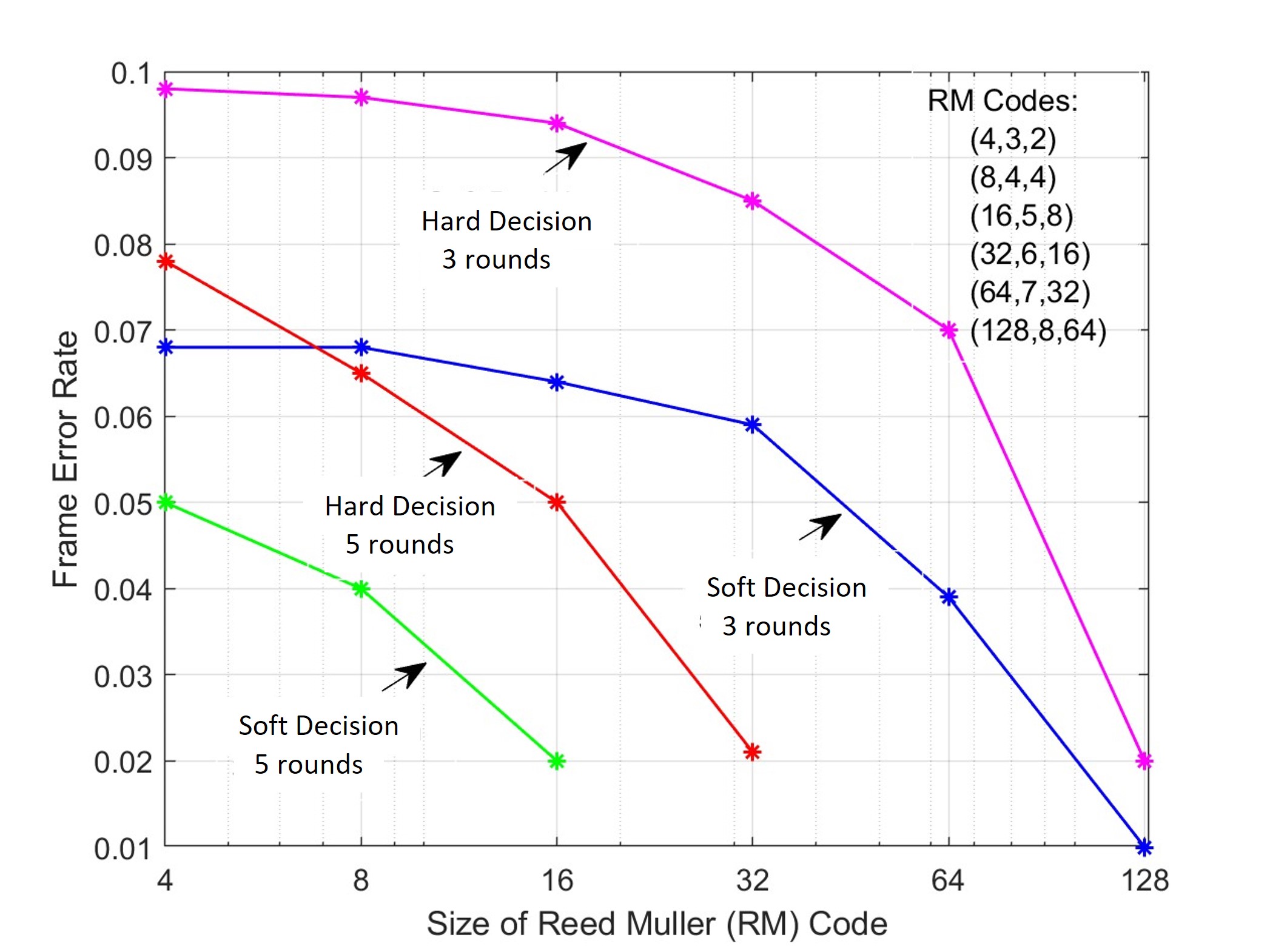}
   \caption{Error Rate of a single code-word for different RM codes.}
 \label{fig3} 
\end{figure}


\newpage 



\begin{thebibliography}{100}

\vspace{-0.12cm}
\bibitem{AK1} 
A. K. Khandani, ``Looping for Encryption Key Generation Over the Internet: A New Frontier in Physical Layer Security,'' {\em 2023 Biennial Symposium on Communications (BSC)}, Montreal, QC, Canada, 2023, pp. 59-64


\vspace{-0.12cm}
\bibitem{common}
R. Ahlswede and I. Csiszar, ``The Role Of Common Randomness In Information Theory And Cryptography, Part 1: Secrecy Constraints," {\em1991 IEEE International Symposium on Information Theory}, Budapest, Hungary, 1991, pp. 265-265
%

\vspace{-0.12cm}
\bibitem{B1} G. Brassard, L. Salvail ``Secret-key reconciliation by public discussion". {\em Workshop on the Theory and Application of Cryptographic Techniques,} Springer. pp. 410–423, 1993. 

\vspace{-0.12cm}
\bibitem{B2}
Kim-Chi Nguyen, Gilles Van Assche, Nicolas J. Cerf,
``Side-Information Coding with Turbo Codes and its Application to Quantum Key Distribution,'' arXiv:cs/0406001, 2004

\vspace{-0.12cm}
\bibitem{B3}
Sungsik Yoon and Jun Heo, ``Efficient information reconciliation with turbo codes over the quantum channel," {\em 2013 International Conference on ICT Convergence,}  Jeju, South Korea, 2013, pp. 1091-1092

\vspace{-0.12cm}
\bibitem{B4}
K. Kasai, R. Matsumoto and K. Sakaniwa, ``Information reconciliation for QKD with rate-compatible non-binary LDPC codes," {\em 2010 International Symposium On Information Theory \& Its Applications}, Taichung, Taiwan, 2010, pp. 922-927

\vspace{-0.12cm}
\bibitem{B5}
J. Martínez-Mateo, D. Elkouss and V. Martín, ``Interactive reconciliation with low-density parity-check codes," {\em 2010 6th International Symposium on Turbo Codes \& Iterative Information Processing}, Brest, France, 2010, pp. 270-274

\vspace{-0.12cm}
\bibitem{B6}
R. Müller, D. Bacco, L. K. Oxenløwe and S. Forchhammer, ``Information Reconciliation for High-Dimensional Quantum Key Distribution using Nonbinary LDPC codes," {\em 2023 12th International Symposium on Topics in Coding}, Brest, France, 2023, pp. 1-5

\vspace{-0.12cm}
\bibitem{B7}
M. Zhu, K. Cui, S. Li, L. Kong, S. Tang and J. Sun, ``A Code Rate-Compatible High-Throughput Hardware Implementation Scheme for QKD Information Reconciliation," {\em Journal of Lightwave Technology}, vol. 40, no. 12, pp. 3786-3793, 15 June15, 2022


\vspace{-0.12cm}
\bibitem{Arxiv}  Amir K. Khandani, ``Quantum-safe Encryption: A New Method to Reduce Complexity and/or
  Improve Security Level,"  arXiv:2401.16302


\vspace{-0.12cm}
\bibitem{R0} McEliece, Robert J. ``A public-key cryptosystem based on algebraic." Coding Theory, 4244 (1978): 114-116.

\vspace{-0.12cm}
\bibitem{Main-ref} Bernstein, D.J., Lange, T., Peters, C. ``Attacking and Defending the McEliece Cryptosystem," {\em Buchmann, J., Ding, J. (eds) Post-Quantum Cryptography. PQCrypto 2008. Lecture Notes in Computer Science, vol 5299. Springer, Berlin, Heidelberg, 2008, pp. 31---46.}




%




%

%


%
%
%
%
%
%
%
%
%
%
%
%

\end{thebibliography}
\end{document}